\documentclass[8pt, a4paper]{article}
\usepackage{amsmath, fullpage}
\usepackage{sectsty}
\usepackage[left=2cm, right=2cm]{geometry}
\sectionfont{\fontsize{12}{12}\selectfont}
\subsectionfont{\fontsize{11}{12}\selectfont}
\usepackage{apacite}
\usepackage{amssymb}
\usepackage{amsthm}
\usepackage{graphicx}
\usepackage{authblk}
\usepackage{algorithm}
\usepackage{algpseudocode}

\usepackage{tikz}
\tikzset{every picture/.style={line width=0.75pt}}

\usepackage{lipsum}
\usepackage{titlesec}
\newtheorem{theorem}{Theorem}[section]

\usepackage{etoolbox}
\AtBeginEnvironment{theorem}{\small}

\usepackage{algorithm}
\usepackage{algpseudocode}

\titlespacing\section{0pt}{12pt plus 4pt minus 2pt}{0pt plus 2pt minus 2pt}
\titlespacing\subsection{0pt}{12pt plus 4pt minus 2pt}{0pt plus 2pt minus 2pt}

\begin{document}
\font\myfont=cmr12 at 15pt
\title{\myfont Over-the-Counter Market Making via Reinforcement Learning}
\author{Zhou Fang \hspace{1cm} Haiqing Xu}
\maketitle

\begin{abstract}
The over-the-counter (OTC) market is characterized by a unique feature that allows market makers to adjust bid-ask spreads based on order size. However, this flexibility introduces complexity, transforming the market-making problem into a high-dimensional stochastic control problem that presents significant challenges. To address this, this paper proposes an innovative solution utilizing reinforcement learning techniques to tackle the OTC market-making problem. By assuming a linear inverse relationship between market order arrival intensity and bid-ask spreads, we demonstrate the optimal policy for bid-ask spreads follows a Gaussian distribution. We apply two reinforcement learning algorithms to conduct a numerical analysis, revealing the resulting return distribution and bid-ask spreads under different time and inventory levels. 
\end{abstract}

\section{Introduction}
Market makers play a crucial role as liquidity providers in diverse financial markets, adapting their liquidity provision strategies to suit the specific characteristics of each market. High-frequency trading (HFT) firms primarily engage in market-making activities, which are widely recognized for their contributions to stabilizing and enhancing the efficiency of financial markets.

OTC markets encompass a wide range of financial instruments such as foreign exchanges (FX), bonds, and stocks that are not traded on formal exchanges for various reasons. In an OTC market, a specific asset class market is typically consisted by multiple dealers-to-clients (D2C) platforms and sometimes supplemented by a dealers-to-dealers (D2D) network, as observed in the case of FX markets. Market makers play a crucial role in the OTC market by providing liquidity to clients through bid-ask prices that they offer for buying and selling. These bid-ask prices are adjusted based on the order size. Given the inherent volatility of financial markets, market makers aim to maintain manageable inventories and achieve this by either adjusting their bid-ask price quotes or utilizing external avenues such as the D2D network to manage inventory levels.

The study of market making originated in the 1980s with seminal works such as \cite{grossman1988liquidity} and \cite{ho1981optimal}. A significant resurgence of interest in market-making occurred with the publication of \cite{avellaneda2008high}, which sparked a wave of subsequent literature in this area. Subsequent works such as \cite{cartea2014buy}, \cite{cartea2015order}, \cite{cartea2017algorithmic}, \cite{cartea2019market}, and \cite{cartea2020market} introduced more sophisticated structures and features, including alpha signals, order flows, and minimum resting time, into market-making models.

However, these aforementioned papers primarily focused on single-dimensional market-making problems and did not explicitly address the challenges posed by multi-dimensional market-making. The curse of dimensionality makes the multi-dimensional market-making problem significantly more intricate. Only recently, a series of papers such as \cite{bergault2021size}, \cite{bergault2021closed}, \cite{barzykin2021market}, and \cite{barzykin2022dealing} have begun to tackle this problem from a purely mathematical perspective.

In parallel, machine learning techniques have also been employed to approach market making in several papers, including \cite{ganesh2019reinforcement}, \cite{beysolow2019market}, \cite{sadighian2020extending}, and \cite{sadighian2020extending}. However, these papers tend to adopt relatively simple models and are primarily focused on demonstrating the application of machine-learning techniques to market-making, rather than delving into the complexities of addressing intricate market-making issues through machine-learning methods.  

In this research paper, we present a novel framework based on reinforcement learning to address the intricate challenges of the multi-dimensional market-making problem encountered by over-the-counter (OTC) market makers. Our approach centers on the utilization of a stochastic policy, which enables the market maker to determine bid-ask spreads. The application of stochastic policies in financial mathematics has gained recognition, as evidenced by prior works such as \cite{wang2020continuous} and \cite{wang2020reinforcement}, which leverage stochastic policies to tackle portfolio management problems. Notably, employing a stochastic policy offers several advantages, including enhanced robustness and the ability to balance exploration and exploitation, as elucidated in the notable contributions of \cite{jia2022policy(a)} and \cite{jia2022policy(b)}, who propose a unified framework encompassing policy evaluation and policy gradient techniques, building upon the earlier aforementioned works.

Within the scope of this paper, we specifically focus on the scenario where market order arrivals follow a Poisson process, with the intensity of arrivals being inversely proportional to bid-ask spreads. Under this assumption, we demonstrate that the optimal policy for bid-ask spread determination can be modeled as a Gaussian distribution. To validate the effectiveness of our proposed framework, we conduct extensive numerical experiments using simulated data, showcasing various performance metrics and bid-ask spread outcomes for specific inventory levels.

The subsequent sections of this paper are structured as follows. Section 2 provides a comprehensive overview of the model setup, elucidating the key components and variables employed in our analysis. In Section 3, we derive the Hamilton-Jacobi-Bellman (HJB) equation and establish the optimal policy's formulation, shedding light on the fundamental principles guiding our approach. Moving forward, Section 4 presents proof of the policy improvement theorem, which serves as the cornerstone of our primary method for approximating the optimal policy.

To validate the efficacy and practicality of our proposed approach, Section 5 presents a detailed analysis of the numerical results obtained through extensive experimentation. In this section, we compare the performance of our method against a traditional actor-critic algorithm, providing insightful comparisons and contrasting the bid-ask spreads policy under specific factors such as inventory level, and asset reference prices.

\section{Model}
This paper delves into the intricate realm of stochastic control problems encountered by OTC market makers, who are confronted with the challenging task of setting diverse bid-ask quotes tailored to orders of varying sizes. Furthermore, these market makers possess the flexibility to externalize their inventory to fellow market participants, thereby introducing an additional dimension of strategic decision-making.  

In the context of our study, we focus on a relatively brief trading period denoted as $[0, T]$ that typically spans several hours or a single trading day. To streamline the problem and facilitate analysis, we adopt the assumption that the volatility of the underlying asset remains constant throughout the entire trading period. This assumption allows us to concentrate on the dynamics of the asset's reference prices, also known as mid-prices, which are defined as follows:
\begin{align}
    \frac{dS_t}{S_t} = \sigma dW_t
\end{align}

In our model, we characterize the occurrence of buy and sell order executions of size $z_k$ as Poisson processes. Specifically, we denote the number of buy order executions as $N_t^+(k)$ and the number of sell order executions as $N_t^-(k)$. The intensities of these Poisson processes are denoted as $\lambda^+(k)$ and $\lambda^-(k)$, respectively.

In this paper, we simplify the setting by assuming that the market maker does not externalize its inventory. As a result, the dynamics of the inventory can be described as follows:
\begin{align}
    dq_t = \sum_{k = 1}^{K} z_k\big(dN_t^+(k) - dN_t^-(k) \big)
\end{align}
in this paper, we assume a simple setting, which is the market orders arrival intensity is linearly inverse to bid-ask spreads, shown as follows,
\begin{align}
    &\lambda_t^\pm (k) = A_{k} - B_{k} \epsilon_t^{a,b}(k) 
\end{align}
Let $\boldsymbol{\epsilon}_t = \left(\epsilon_t^{b, a}(k)\right)_{k=1}^{K}$ represent the bid-ask spreads posted by the market maker at time $t$. The probability density function for posting spreads $\boldsymbol{\epsilon}_t$ is denoted as $\boldsymbol{\pi}(\boldsymbol{\epsilon}_t | t, S, q)$. If the market maker chooses to post the spreads $\boldsymbol{\epsilon}_t$ at time $t$, the wealth exhibits the following dynamics:
\begin{align}
    dX_t = \sum_{k = 1}^{K} z_k \big[ \epsilon_t^b(k)dN_t^+(k) + \epsilon_t^a(k)dN_t^-(k)\big] + d(q_t S_t)
\end{align}

\section{Market-Making in OTC Markets}
\subsection{Hamilton-Jacobian-Bellman Equation}
 
In the context of our study, let us consider a policy $\boldsymbol{\pi}$ that guides the market maker's actions. We denote the corresponding inventory process under policy $\boldsymbol{\pi}$ as $q^{\boldsymbol{ \pi}}_t$, where its initial condition at time $t$ is specified as $S_t = S$ and $q^{\boldsymbol\pi}_t = q$. The central objective of our market maker is to optimize the expected profit by effectively managing inventory risk and promoting exploration through the utilization of a stochastic policy. To quantitatively assess the level of exploration, we employ the cross-entropy of the stochastic policy, which offers insights into the policy's capacity for exploration and exploitation. Furthermore, to discourage excessive inventory holdings over the trading period, we introduce a penalty term based on the average square of the inventory. By combining these elements, we define the value function under policy $\boldsymbol{\pi}$ as follows:
 \begin{align}
     &\hspace{0.5cm} V^{\boldsymbol \pi}(t, S, q) \nonumber \\
     &=\mathbb{E} \bigg[ \int_t^T \int_{\boldsymbol \epsilon_u} \Big[  \sum_{k = 1}^{K} z_k\big[ \epsilon_u^b(k)dN_u^+(k) + \epsilon_u^a(k)dN_u^-(k)\big] + d(q_u S_u)  \Big] \boldsymbol \pi (\boldsymbol \epsilon_u | u, S_u, q_u^{\boldsymbol \pi}) d\boldsymbol \epsilon_u  \nonumber \\
     &- \gamma \int_t^T \int_{\boldsymbol \epsilon_u}\boldsymbol \pi(\boldsymbol \epsilon_u | u, S_u, q_u^{\boldsymbol \pi}) \log \boldsymbol \pi(\boldsymbol \epsilon_u | u, S_u, q_u^{\boldsymbol \pi}) d\boldsymbol \epsilon_u du  - \delta \int_t^T q_u^2 du \hspace{0.1cm} \bigg | \hspace{0.1cm} S_t = S, q_t^{\boldsymbol \pi} = q \bigg] \nonumber \\
     &= \mathbb{E} \bigg[ \int_t^T \int_{\boldsymbol \epsilon_u} \sum_{k = 1}^{K} \Big[ z_k \big( S_u + \epsilon_u^b(k) \big) dN_u^+(k) - z_k\big( S_u - \epsilon_u^a(k) \big) dN_u^-(k) \Big] \boldsymbol \pi(\boldsymbol \epsilon_u | u, S_u, q_u^{\boldsymbol \pi}) d\boldsymbol \epsilon_u \nonumber \\
     & - \gamma \int_t^T \int_{\boldsymbol \epsilon_u}\boldsymbol \pi(\boldsymbol \epsilon_u | u, S_u,  q_u^{\boldsymbol \pi}) \log \boldsymbol \pi(\boldsymbol \epsilon_u | u, S_u,  q_u^{\boldsymbol \pi}) d\boldsymbol \epsilon_u du   - \delta \int_t^T q_u^2 du\hspace{0.1cm} \bigg | \hspace{0.1cm} S_t = S, q_t^{\boldsymbol \pi} = q \bigg]
 \end{align}
 then the value function under the optimal policy is 
 \begin{align}
     &\hspace{0.5cm}V(t, S, q) \nonumber \\
     & = \underset{\boldsymbol \pi}{\max} \hspace{0.1cm} \mathbb{E} \bigg[ \int_t^T \int_{\boldsymbol \epsilon_u} \sum_{k = 1}^{K} \Big[ z_k \big( S_u + \epsilon_u^b(k) \big) dN_u^+(k) - z_k\big( S_u - \epsilon_u^a(k) \big) dN_u^-(k) \Big] \boldsymbol \pi(\boldsymbol \epsilon_u | u, S_u,  q_u^{\boldsymbol \pi}) d\boldsymbol \epsilon_u \nonumber \\
     & - \gamma \int_t^T \int_{\boldsymbol \epsilon_u}\boldsymbol \pi(\boldsymbol \epsilon_u | u, S_u,  q_u^{\boldsymbol \pi}) \log \boldsymbol \pi(\boldsymbol \epsilon_u | u, S_u,  q_u^{\boldsymbol \pi}) d\boldsymbol \epsilon_u du   - \delta \int_t^T q_u^2 du \hspace{0.1cm} \bigg | \hspace{0.1cm} S_t = S, q_t^{\boldsymbol \pi} = q \bigg] \nonumber \\
     &= \underset{\boldsymbol \pi}{\max} \hspace{0.1cm} \mathbb{E} \bigg[ \int_t^{t + \Delta t} \int_{\boldsymbol \epsilon_u} \sum_{k = 1}^{K} \Big[ z_k \big( S_u + \epsilon_u^b(k)  \big) dN_u^+(k) - z_k\big( S_u - \epsilon_u^a(k) \big) dN_u^-(k) \Big] \boldsymbol \pi(\boldsymbol \epsilon_u | u, S_u,  q_u^{\boldsymbol \pi}) d\boldsymbol \epsilon_u - \delta \int_t^{t+\Delta t} q_u^2 du  \nonumber \\
     & - \gamma \int_t^{t + \Delta t} \int_{\boldsymbol \epsilon_u}\boldsymbol \pi(\boldsymbol \epsilon_u | u, S_u, q_u^{\boldsymbol \pi}) \log \boldsymbol \pi(\boldsymbol \epsilon_u | u, S_u,  q_u^{\boldsymbol \pi}) d\boldsymbol \epsilon_u du  + V(t + \Delta t, S_t + \Delta S_t, q_t + \Delta q_t^{\boldsymbol \pi})   \hspace{0.1cm} \bigg | \hspace{0.1cm} S_t = S, q_t^{\boldsymbol \pi} = q \bigg] \nonumber \\
     & = \underset{\boldsymbol \pi}{\max} \hspace{0.1cm} \bigg \{  \int_{\boldsymbol \epsilon_t} \sum_{k = 1}^{K} \Big[ z_k \big( S_t + \epsilon_t^b(k) \big) \lambda_t^+(k) - z_k\big( S_t - \epsilon_t^a(k) \big) dN_t^-(k) \Big] \boldsymbol \pi(\boldsymbol \epsilon_t | t, S_t,  q_t^{\boldsymbol \pi}) d\boldsymbol \epsilon_t  \Delta t - \delta q_t^2\Delta t \nonumber \\
     & - \gamma \int_{\boldsymbol \epsilon_t}\boldsymbol \pi(\boldsymbol \epsilon_t | t, S_t, q_t^{\boldsymbol \pi}) \log \boldsymbol \pi(\boldsymbol \epsilon_t | t, S_t, q_t^{\boldsymbol \pi}) d\boldsymbol \epsilon_t \Delta t  + \mathbb{E} \Big[   V(t + \Delta t, S_t + \Delta S_t, q_t^{\boldsymbol \pi} + \Delta q_t^{\boldsymbol \pi})  \hspace{0.1cm} \Big | \hspace{0.1cm} S_t = S, q_t^{\boldsymbol \pi} = q  \Big ]
     \bigg \}
 \end{align}
to make notation simplier, denote $\mathcal{L}V(t, S_t, q_t)$ as
\begin{align}
    \mathcal{L}V(t, S_t, q_t) = V(t, S_t, q_t) + \big(\partial_t V(t, S_t, q_t) + \frac{1}{2} \sigma^2 \partial_{SS} V(t, S_t, q_t)\big) \Delta t + \sigma \partial_S V(t, S_t, q_t) dW_t
\end{align}
since $dS_t = \sigma S_t dW_t$, and $dq_t = \sum_k z_k(dN_t^+(k) - dN_t^-(k))$, by the Ito formula, we have the following,   
\begin{align}
    & \hspace{0.5cm}V(t + \Delta t, S_t + \Delta S_t, q_t + \Delta q_t) \nonumber \\
    &= V(t + \Delta t, S_t + \Delta S_t, q_t) \prod_{k} (1 - dN_t^+(k))(1 - dN_t^-(k)) \nonumber \\
    & + \sum_{k} \Big[V(t + \Delta t, S_t +  \Delta S_t, q_t + z_k) dN_t^+(k) + V(t + \Delta t, S_t +  \Delta S_t, q_t - z_k) dN_t^-(k)\Big] \\
    &= \mathcal{L}V(t, S_t, q_t) \prod_{k} (1 - dN_t^+(k))(1 - dN_t^-(k)) + \sum_{k} \mathcal{L}V(t, S_t, q_t + z_k)dN_t^+(k) + \mathcal{L}V(t, S_t, q_t - z_k) dN_t^-(k)  \nonumber
\end{align}
it is crucial to note that the previous derivation of the Ito formula is based on the assumption that the inventory process is defined as $dq_t = \sum_k z_k(dN_t^+(k) - dN_t^-(k))$. It is essential to emphasize that the intensities of the Poisson processes are determined by the quoted bid-ask spreads. As a result, the inventory process in the aforementioned Ito formula assumes that the bid-ask spreads are already determined. Hence, when calculating the conditional expectation $\mathbb{E}[V(t+\Delta t, S_t + \Delta S_t, q_t^{\boldsymbol{\pi}} + \Delta q_t^{\boldsymbol{\pi}}) | S_t = S, q_t^{\boldsymbol{\pi}} = q]$, it becomes necessary to consider and average over all possible scenarios. Consequently, the conditional expectation can be expressed as follows, taking into account the various scenarios
\begin{align}
    &\hspace{0.5cm}\mathbb{E} \Big[ V(t+\Delta t, S_t + \Delta S_t, q_t^{\boldsymbol \pi} + \Delta q_t^{\boldsymbol \pi}) \hspace{0.1cm} \Big | \hspace{0.1cm} S_t = S, q_t^{\boldsymbol \pi} = q \Big] \nonumber \\
    &= V(t, S_, q) + \int_{\boldsymbol \epsilon_t} \boldsymbol \pi(\boldsymbol \epsilon_t | t, S, q)\Big[ - \sum_{k} \big(\lambda_t^+(k)+ \lambda_t^-(k)\big) V(t, S, q) + \partial_t V(t, S, q) + \frac{1}{2} \sigma^2 \partial_{SS} V(t, S, q) \nonumber \\
    &\hspace{0.3cm}+ \sum_{k} \big[\lambda_t^+(k) V(t, S, q + z_k) + \lambda_t^-(k) V(t, S, q - z_k) \big]  \Big] d \boldsymbol \epsilon_t \Delta t
\end{align}
then the Hamilton-Jacobian-Bellman equation is 
\begin{align}
    & \underset{\boldsymbol \pi}{\max} \Bigg\{ \hspace{0.1cm} \int_{\boldsymbol \epsilon_t} \sum_{k} \Big[\lambda_t^+(k) V(t, S, q + z_k) + \lambda_t^-(k) V(t, S, q - z_k) - \big(\lambda_t^+(k)+ \lambda_t^-(k)\big) V(t, S, q) \Big]   \boldsymbol \pi(\boldsymbol \epsilon_t | t, S, q)d\boldsymbol \epsilon_t \nonumber \\
    & +  \int_{\boldsymbol \epsilon_t} \sum_{k = 1}^{N} \Big[ z_k \lambda_t^+(k)\big( S + \epsilon_t^b(k) \big) - z_k \lambda_t^-(k)\big( S - \epsilon_t^a(k) \big) \Big] \boldsymbol \pi(\boldsymbol \epsilon_t | t, S, q) d\boldsymbol \epsilon_t \nonumber \\
     & - \gamma \int_{\boldsymbol \epsilon_t}\boldsymbol \pi(\boldsymbol \epsilon_t | t, S,  q) \log \boldsymbol \pi(\boldsymbol \epsilon_t | t, S, q) d\boldsymbol \epsilon_t \Bigg \} - \delta q_t^2 + \partial_t V(t, S, q) + \frac{1}{2} \sigma^2 \partial_{SS} V(t, S, q) \nonumber \\
     & = 0
\end{align}

\subsection{Optimal Stochastic Policy}
To determine the maximizer for the quantity inside the max bracket of the HJB equation, we utilize the calculus of variations
\begin{align}
    0 = &\int_{\boldsymbol \epsilon_t} \sum_{k} \Big[\lambda_t^+(k) V(t, S, q + z_k) + \lambda_t^-(k) V(t, S, q - z_k) - \big(\lambda_t^+(k)+ \lambda_t^-(k)\big) V(t, S, q) \Big]  \delta \boldsymbol \pi d\boldsymbol \epsilon_t \nonumber \\
    & +  \int_{\boldsymbol \epsilon_t} \sum_{k = 1}^{N} \Big[ z_k \lambda_t^+(k)\big( S + \epsilon_t^b(k)  \big) - z_k \lambda_t^-(k)\big( S - \epsilon_t^a(k)  \big) \Big] \delta \boldsymbol \pi d\boldsymbol \epsilon_t \nonumber \\
     & - \gamma \int_{\boldsymbol \epsilon_t}\boldsymbol \pi \frac{\delta \boldsymbol \pi}{\boldsymbol \pi} d\boldsymbol \epsilon_t - \gamma \int_{\boldsymbol \epsilon_t} \delta \boldsymbol \pi \log \boldsymbol \pi d \boldsymbol \epsilon_t
\end{align}
since $\boldsymbol \pi$ is probability density distribution, then 
\begin{align}
    \int_{\boldsymbol \epsilon_t} \delta \boldsymbol \pi d \boldsymbol \epsilon_t = 0
\end{align}
the equation (10) becomes
\begin{align}
    0 = &\int_{\boldsymbol \epsilon_t} \delta \boldsymbol \pi \bigg (\sum_{k} \Big[\lambda_t^+(k) V(t, S, q + z_k) + \lambda_t^-(k) V(t, S, q - z_k) - \big(\lambda_t^+(k)+ \lambda_t^-(k)\big) V(t, S_t, q_t) \nonumber \\
    & +  z_k \lambda_t^+(k)\big( S + \epsilon_t^b(k)  \big) - z_k \lambda_t^-(k)\big( S - \epsilon_t^a(k)  \big) \Big] - \gamma 
 (\delta \boldsymbol \pi) \log \boldsymbol \pi   d\boldsymbol \epsilon_t \bigg)   d \boldsymbol \epsilon_t
\end{align}
the quantity inside the bracket above is a constant
\begin{align}
    C &= \sum_{k} \Big[\lambda_t^+(k) V(t, S, q + z_k) + \lambda_t^-(k) V(t, S, q - z_k) - \big(\lambda_t^+(k)+ \lambda_t^-(k)\big) V(t, S, q) \nonumber \\
    &+  z_k \lambda_t^+(k)\big( S + \epsilon_t^b(k) \big) - z_k \lambda_t^-(k)\big( S - \epsilon_t^a(k) \big) \Big] - \gamma \log \boldsymbol \pi 
\end{align}
to simplify the notations, let
\begin{align}
    &\mathcal{H}_k^+(t, S, q, \boldsymbol \pi) = V^{\boldsymbol \pi}(t, S, q + z_k) - V^{\boldsymbol \pi}(t, S, q) + z_k S \\
    &\mathcal{H}_k^-(t, S, q, \boldsymbol \pi) = V^{\boldsymbol \pi}(t, S, q - z_k) - V^{\boldsymbol \pi}(t, S, q) - z_k S
\end{align}
under the optimal policy $\boldsymbol \pi^*$, there is 
\begin{align}
    &\mathcal{H}_k^+(t, S, q) = V(t, S, q + z_k) - V(t, S, q) + z_k S  \\
    &\mathcal{H}_k^-(t, S, q) = V(t, S, q - z_k) - V(t, S, q) - z_k S 
\end{align}
then the optimal stochastic policy is  
\begin{align}
    \boldsymbol \pi^* (\boldsymbol \epsilon_t | t, S, q) &\propto \exp \Big \{ \frac{1}{\gamma}  \sum_k \big(A_k - B_k \epsilon_t^{a,b}(k)\big) \big(z_k \epsilon_t^{a,b}(k) + \mathcal{H}_k^\pm(t, S, q)\big)  \Big \} \nonumber \\
    &\propto \prod_k \exp \Big \{ -\frac{z_kB_k}{\gamma} \Big[  \epsilon_t^{a,b}(k) - \frac{A_k}{2B_k} + \frac{\mathcal{H}_k^\pm(t, S, q)}{2z_k}    \Big]^2  \Big \} \nonumber \\
    &\propto \prod_k \mathcal{N} \Big(\epsilon_t^{a,b} \hspace{0.1cm} \big| \hspace{0.1cm}    \frac{A_k}{2B_k} - \frac{\mathcal{H}_k^\pm(t, S, q)}{2z_k}, \frac{\gamma}{2z_kB_k}       \Big)
\end{align}
hence, it is evident that the optimal policy corresponds to a multi-dimensional Gaussian distribution. To streamline the notation, let us introduce the following notation 
\begin{align}
    \boldsymbol \mu(t, S, q, \boldsymbol \pi) = \Big(\frac{A_1}{2B_1} - \frac{\mathcal{H}_1^\pm(t, S, q, \boldsymbol \pi)}{2z_1}, ..., \frac{A_N}{2B_N} - \frac{\mathcal{H}_N^\pm(t, S, q, \boldsymbol \pi)}{2z_N} \Big) \nonumber 
\end{align}

\[
  \boldsymbol \Sigma =
  \begin{bmatrix}
    \frac{\gamma}{2z_1B_1} & & & &  \\
    & \frac{\gamma}{2z_1B_1} & & &  \\
    & & \ddots & &\\
    & & &\frac{\gamma}{2z_kB_k} & \\
    & & & & \frac{\gamma}{2z_kB_k} 
  \end{bmatrix}
\]
then the optimal policy becomes
\begin{align}
    \boldsymbol \pi^* \sim \mathcal{N}(\cdot \hspace{0.1cm}| \hspace{0.1cm} \boldsymbol \mu(t, S, q, \boldsymbol \pi^*), \boldsymbol \Sigma)
\end{align}

\section{Policy Improvement Theorem}
\begin{theorem}[policy improvement theorem]
Given any $\boldsymbol \pi$, let the new policy $\boldsymbol \pi_{new}$ to be
\begin{align}
    \boldsymbol  \pi_{new}  \sim \mathcal{N} (\hspace{0.1cm} \cdot \hspace{0.1cm}| \hspace{0.1cm} \boldsymbol \mu(t, S, q, \boldsymbol \pi), \boldsymbol \Sigma )
\end{align}
then the following inequality holds
\begin{align}
    V^{\boldsymbol \pi}(t, S, q) \leq V^{\boldsymbol \pi_{new}}(t, S,  q)
\end{align}
\end{theorem}
\begin{proof}
Let $q_t^{\boldsymbol \pi_{new}}$ denote the inventory process under the policy $\pi_{new}$, with an initial condition at time $t$ of $q_t^{\boldsymbol \pi_{new}} = q$, and $S_t = S$. Applying the Ito formula and averaging over all possible scenarios, we obtain the following expression
\begin{align}
    &\hspace{0.5cm}V^{\boldsymbol \pi}(t, S, q) \nonumber \\
    & = \mathbb{E} \Big[V^{\boldsymbol \pi}(s, S_s, q_s^{\boldsymbol \pi_{new}}) + \int_t^s \int_{\boldsymbol \epsilon_u}  \boldsymbol \pi_{new}(\boldsymbol \epsilon_u | u, S_u, \boldsymbol q_u^{\boldsymbol \pi_{new}}) V^{\boldsymbol \pi}(u, S_u, q_u^{ \boldsymbol \pi_{new}})\sum_k [\lambda_u^+(k) + \lambda_u^-(k)] d \boldsymbol \epsilon_u du \nonumber \\
    & -\int_t^s \int_{\boldsymbol \epsilon_u} \boldsymbol \pi_{new}(\boldsymbol \epsilon_u | u, S_u, q_u^{\boldsymbol \pi_{new}})\sum_{k}\Big[V^{\boldsymbol \pi}(u, S_u, q_u^{\boldsymbol \pi_{new}} + z_k) \lambda_u^+(k)  + V^{\boldsymbol \pi}(u, S_u, q_u^{\boldsymbol \pi_{new}} - z_k) \lambda_u^-(k)  \Big] d\boldsymbol \epsilon_u du  \nonumber  \\
    & - \int_s^t \Big( \partial_t V^{\boldsymbol \pi}(u, S_u, q_u^{\boldsymbol \pi_{new}}) + \frac{1}{2} \sigma^2 \partial_{SS} V^{\boldsymbol \pi} (u, S_u, q_u^{\boldsymbol \pi_{new}}) \Big) du \hspace{0.1cm}\Big| \hspace{0.1cm} S_t = S, q_t^{\boldsymbol \pi_{new}} = q \Big]  
\end{align}
since at time $t$, under policy $\boldsymbol \pi$, the following equality holds 
\begin{align}
    &  \int_{\boldsymbol \epsilon_t} \sum_{k} \Big[\lambda_t^+(k) V^{\boldsymbol \pi}(t, S, q + z_k) + \lambda_t^-(k) V^{\boldsymbol \pi}(t, S, q - z_k) - \big(\lambda_t^+(k)+ \lambda_t^-(k)\big) V^{\boldsymbol \pi}(t, S, q) \Big]   \boldsymbol \pi(\boldsymbol \epsilon_t | t, S, q)d\boldsymbol \epsilon_t \nonumber \\
    & +  \int_{\boldsymbol \epsilon_t} \sum_{k = 1}^{N} \Big[ z_k \lambda_t^+(k)\big( S + \epsilon_t^b(k) \big) - z_k \lambda_t^-(k)\big( S - \epsilon_t^a(k)\big) \Big] \boldsymbol \pi(\boldsymbol \epsilon_t | t, S, q) d\boldsymbol \epsilon_t \nonumber \\
     & - \gamma \int_{\boldsymbol \epsilon_t}\boldsymbol \pi(\boldsymbol \epsilon_t | t, S, q) \log \boldsymbol \pi(\boldsymbol \epsilon_t | t, S, q) d\boldsymbol \epsilon_t - \delta q_t^2  + \partial_t V^{\boldsymbol \pi}(t, S, q) + \frac{1}{2} \sigma^2 \partial_{SS} V^{\boldsymbol \pi}(t, S, q)  \nonumber \\
     & = 0
\end{align}
for the policy $\boldsymbol \pi_{new}$, given its construction, and employing the same calculus of variation arguments as in equations $(10) - (13)$, we can conclude that $\boldsymbol \pi_{new}$ maximizes the following quantity
\begin{align}
    & \underset{\widetilde{\boldsymbol \pi}}{\max} \hspace{0.1cm} \bigg \{\int_{\boldsymbol \epsilon_t} \sum_{k = 1}^{N} \Big[ z_k \lambda_t^+(k)\big( S + \epsilon_t^b(k)  \big) - z_k \lambda_t^-(k)\big( S - \epsilon_t^a(k)  \big) \Big] \widetilde{\boldsymbol \pi}(\boldsymbol \epsilon_t | t, S, q) d\boldsymbol \epsilon_t \nonumber \\
    &+ \int_{\boldsymbol \epsilon_t} \sum_{k} \Big[\lambda_t^+(k) V^{\boldsymbol \pi}(t, S, q + z_k) + \lambda_t^-(k) V^{\boldsymbol \pi}(t, S, q - z_k) - \big(\lambda_t^+(k)+ \lambda_t^-(k)\big) V^{\boldsymbol \pi}(t, S, q) \Big]   \widetilde{\boldsymbol \pi}(\boldsymbol \epsilon_t | t, S, q)d\boldsymbol \epsilon_t \nonumber \\
     & - \gamma \int_{\boldsymbol \epsilon_t} \widetilde{\boldsymbol \pi}(\boldsymbol \epsilon_t | t, S, q) \log \widetilde{\boldsymbol \pi}(\boldsymbol \epsilon_t | t, S, q) d\boldsymbol \epsilon_t \bigg \}
\end{align}
which results in the following inequality 
\begin{align}
        & \int_{\boldsymbol \epsilon_t} \sum_{k} \Big[\lambda_t^+(k) V^{\boldsymbol \pi}(t, S, q + z_k) + \lambda_t^-(k) V^{\boldsymbol \pi}(t, S, q - z_k) - \big(\lambda_t^+(k)+ \lambda_t^-(k)\big) V^{\boldsymbol \pi}(t, S, q) \Big]   \boldsymbol \pi_{new}(\boldsymbol \epsilon_t | t, S, q)d\boldsymbol \epsilon_t \nonumber \\
    & +  \int_{\boldsymbol \epsilon_t} \sum_{k = 1}^{N} \Big[ z_k \lambda_t^+(k)\big( S + \epsilon_t^b(k)  \big) - z_k \lambda_t^-(k)\big( S - \epsilon_t^a(k) \big) \Big] \boldsymbol \pi_{new}(\boldsymbol \epsilon_t | t, S, q) d\boldsymbol \epsilon_t \nonumber \\
     & - \gamma \int_{\boldsymbol \epsilon_t}\boldsymbol \pi_{new}(\boldsymbol \epsilon_t | t, S, q) \log \boldsymbol \pi_{new}(\boldsymbol \epsilon_t | t, S, q) d\boldsymbol \epsilon_t  - \delta q_t^2  + \partial_t V^{\boldsymbol \pi}(t, S, q) + \frac{1}{2} \sigma^2 \partial_{SS} V^{\boldsymbol \pi}(t, S, q)  \nonumber \\
     & \geq 0
\end{align}
then equation $(23)$ yields
\begin{align}
    &\hspace{0.5cm}V^{\boldsymbol \pi}(t, S, q) \nonumber \\
    & \leq \mathbb{E} \Big[ \int_t^s\int_{\boldsymbol \epsilon_t} \sum_{k = 1}^{N} \Big[ z_k \lambda_u^+(k)\big( S_u + \epsilon_u^b(k) \big) - z_k \lambda_u^-(k)\big( S_u - \epsilon_u^a(k)  \big) \Big] \boldsymbol \pi_{new}(\boldsymbol \epsilon_u | u, S_u, q_u^{\boldsymbol \pi_{new}}) d\boldsymbol \epsilon_u du  - \delta \int_t^s q_u^2 du \nonumber \\ 
     & - \gamma \int_t^s \int_{\boldsymbol \epsilon_u} \boldsymbol \pi_{new}(\boldsymbol \epsilon_u | u, S_u, q_u^{\boldsymbol \pi_{new}}) \log  \boldsymbol \pi_{new}(\boldsymbol \epsilon_u | u, S_u, q_u^{\boldsymbol \pi_{new}}) d\boldsymbol\epsilon_u du + V^{\boldsymbol \pi}(s, S_s, q_s^{\boldsymbol \pi_{new}})  \hspace{0.1cm} \Big | \hspace{0.1cm} S_t = S, q_t^{\boldsymbol \pi_{new}} = q \Big] 
\end{align}
by setting $s = T$, we obtain $V^{\boldsymbol \pi}(T, S_T, q_T^{\boldsymbol \pi_{new}}) = V^{\boldsymbol \pi_{new}}(T, S_T, q_T^{\boldsymbol \pi_{new}})$. Substituting this into equation $(75)$, there is
\begin{align}
    &\hspace{0.5cm}V^{\boldsymbol \pi}(t, S, q) \nonumber \\
    & \leq \mathbb{E} \Big[ \int_t^T\int_{\boldsymbol \epsilon_t} \sum_{k = 1}^{N} \Big[ z_k \lambda_u^+(k)\big( S_u + \epsilon_u^b(k) \big) - z_k \lambda_u^-(k)\big( S_u - \epsilon_u^a(k) \big) \Big] \boldsymbol \pi_{new}(\boldsymbol \epsilon_u | u, S_u, q_u^{\boldsymbol \pi_{new}}) d\boldsymbol \epsilon_u du - \delta \int_t^T q_u^2 du\nonumber \\ 
     & - \gamma \int_t^T \int_{\boldsymbol \epsilon_u} \boldsymbol \pi_{new}(\boldsymbol \epsilon_u | u, S_u, q_u^{\boldsymbol \pi_{new}}) \log  \boldsymbol \pi_{new}(\boldsymbol \epsilon_u | u, S_u, q_u^{\boldsymbol \pi_{new}}) d\boldsymbol\epsilon_u du  + V^{\boldsymbol \pi}(T, S_T, q_T^{\boldsymbol \pi_{new}})  \hspace{0.1cm} \Big | \hspace{0.1cm} S_t = S, q_t^{\boldsymbol \pi_{new}} = q \Big]  \nonumber \\
    &\leq V^{\boldsymbol \pi_{new}} (t, S, q)
\end{align}
\end{proof}

\section{Reinforcement Learning Algorithm}
In this section, we introduce two reinforcement learning algorithms. The first algorithm is based on the policy improvement theorem established earlier, and it utilizes a single neural network to model the value function. On the other hand, the second algorithm follows the actor-critic approach, where separate neural networks are employed to model the policy and the value function. This actor-critic algorithm offers the advantage of better control over the range of bid-ask spreads and typically requires less training time compared to the first algorithm
\subsection{Policy Iteration Algorithm}
Given a policy $\boldsymbol \pi$, and $q_t^{\boldsymbol \pi}$ is inventory process under policy $\boldsymbol \pi$. Let the initial condition at time $t$ to be $S_t = S$, $q_t^{\boldsymbol \pi} = q$, the value function under policy $\boldsymbol \pi$ is
\begin{align}
    &\hspace{0.5cm} V^{\boldsymbol \pi}(t, S, q) \nonumber \\
     &= \mathbb{E} \bigg[ \int_t^s \int_{\boldsymbol \epsilon_u} \sum_{k = 1}^{N} \Big[ z_k \big( S_u + \epsilon_u^b(k)  \big) dN_u^+(k) - z_k\big( S_u - \epsilon_u^a(k)  \big) dN_u^-(k) \Big] \boldsymbol \pi(\boldsymbol \epsilon_u | u, S_u, q_u^{\boldsymbol \pi}) d\boldsymbol \epsilon_u  - \delta \int_t^s q_u^2 du \nonumber \\
     & - \gamma \int_t^s \int_{\boldsymbol \epsilon_u}\boldsymbol \pi(\boldsymbol \epsilon_u | u, S_u, q_u^{\boldsymbol \pi}) \log \boldsymbol \pi(\boldsymbol \epsilon_u | u, S_u, q_u^{\boldsymbol \pi}) d\boldsymbol \epsilon_u du + V(s, S_s, q_s^{\boldsymbol \pi}) \hspace{0.1cm} \bigg | \hspace{0.1cm} S_t = S,  q_t^{\boldsymbol \pi} = q \bigg]
\end{align}
then there is 
\begin{align}
     &\mathbb{E} \bigg[ \frac{1}{s-t}\int_t^s \int_{\boldsymbol \epsilon_u} \sum_{k = 1}^{N} \Big[ z_k \big( S_u + \epsilon_u^b(k)  \big) dN_u^+(k) - z_k\big( S_u - \epsilon_u^a(k)  \big) dN_u^-(k) \Big] \boldsymbol \pi(\boldsymbol \epsilon_u | u, S_u, q_u^{\boldsymbol \pi}) d\boldsymbol \epsilon_u - \frac{\delta}{s - t} \int_t^s q_u^2 du \nonumber \\
     & - \frac{\gamma}{s-t} \int_t^s \int_{\boldsymbol \epsilon_u}\boldsymbol \pi(\boldsymbol \epsilon_u | u, S_u, q_u^{\boldsymbol \pi}) \log \boldsymbol \pi(\boldsymbol \epsilon_u | u, S_u, q_u^{\boldsymbol \pi}) d\boldsymbol \epsilon_u du    + \frac{V^{\boldsymbol \pi}(s, S_s, q_s^{\boldsymbol \pi}) - V^{\boldsymbol \pi}(t, S_t, q_t^{\boldsymbol \pi})}{s-t} \hspace{0.1cm} \bigg | \hspace{0.1cm} S_t = S, q_t^{\boldsymbol \pi} = q \bigg] = 0
\end{align}

By taking the limit as $s$ approaches $t$ and parametrizing the value function $V^{\boldsymbol \pi}_{\theta}$, we can define the temporal difference error as follows
\begin{align}
    \delta_t^{\theta} &= \underset{s \to t}{\lim} \hspace{0.1cm}\mathbb{E} \Big[  \frac{V_{\theta}^{\boldsymbol \pi}(s, S_s, q_s^{\boldsymbol \pi}) - V_{\theta}^{\boldsymbol \pi}(t, S_t, q_t^{\boldsymbol \pi})}{s-t} \hspace{0.1cm} \bigg | \hspace{0.1cm} S_t = S, q_t^{\boldsymbol \pi} = q  \Big] - \gamma \int_{\boldsymbol \epsilon_t}\boldsymbol \pi(\boldsymbol \epsilon_t | t, S, q) \log \boldsymbol \pi(\boldsymbol \epsilon_t | t, S, q) d\boldsymbol \epsilon_t   \nonumber \\
    &+ \int_{\boldsymbol \epsilon_t} \sum_{k = 1}^{N} \Big[ z_k \big( S + \epsilon_t^b(k)  \big) dN_t^+(k) - z_k\big( S - \epsilon_t^a(k) \big) dN_t^-(k) \Big] \boldsymbol \pi(\boldsymbol \epsilon_t | t, S, q) d\boldsymbol \epsilon_t - \delta q_t^2 
\end{align}
using the Monte Carlo method to generate a set of sample paths $\mathcal{D} = \{(t_i, S_i^{d}, q_i^{d}, \Delta N_{t_i}^+, \Delta N_{t_i}^-)_{i = 1}^{T} \}_{d = 1}^{D}$, then the loss to be minimized is
\begin{align}
    \textbf{ML}(\theta) &= \frac{1}{2} \sum_{\mathcal{D}} \sum_i \bigg( \frac{V^{\boldsymbol \pi}_{\theta} (t_{i + 1}, S_{t_{i + 1}}^d, q_{t_{i + 1}}^d) - V_{\theta}^{\boldsymbol \pi}(t_i, S_{t_i}^d, q_{t_i}^d)}{\Delta t}   - \gamma \int_{\boldsymbol \epsilon_{t_i}}\boldsymbol \pi(\boldsymbol \epsilon_{t_i} | t_i, S_{t_i}^d, q_{t_i}^d) \log \boldsymbol \pi(\boldsymbol \epsilon_{t_i} | t_i, S_{t_i}^d, q_{t_i}^d) d\boldsymbol \epsilon_{t_i}    \nonumber \\   
    & \int_{\boldsymbol \epsilon_{t_i}} \sum_{k = 1}^{N} \Big[ z_k \big( S_{t_i}^d + \epsilon_{t_i}^b(k) \big) dN_{t_i}^+(k) - z_k\big( S_{t_i}^d - \epsilon_{t_i}^a(k) \big) dN_{t_i}^-(k) \Big] \boldsymbol \pi(\boldsymbol \epsilon_{t_i} | t_i, S_{t_i}^d, q_{t_i}^d) d\boldsymbol \epsilon_{t_i} - \delta q_{t_i}^2
    \bigg)^2 \Delta t \nonumber \\
    &= \frac{1}{2} \sum_{\mathcal{D}} \sum_i \bigg( \frac{V^{\boldsymbol \pi}_{\theta} (t_{i + 1}, S_{t_{i + 1}}^d, q_{t_{i + 1}}^d) - V_{\theta}^{\boldsymbol \pi}(t_i, S_{t_i}^d, q_{t_i}^d)}{\Delta t}   - \gamma \Big(  N(1 + \log 2\pi) + \sum_k \log \frac{\gamma}{2z_kB_k}   \Big)    \nonumber \\  
    & + \sum_k z_k S_{t_i}^d  \big(\Delta N_{t_i}^+(k) - \Delta N_{t_i}^-(k)\big)  - \delta q_{t_i}^2 \nonumber \\
    &+ \sum_k \Big(\big(\frac{A_k}{2B_k} - \frac{\mathcal{H}^+_\theta(t_i, S_{t_i}, q_{t_i}, \boldsymbol \pi)}{2z_k} \big) \Delta N^+_{t_i}(k) + \big(\frac{A_k}{2B_k} - \frac{\mathcal{H}^-_\theta(t_i, S_{t_i}, q_{t_i}, \boldsymbol \pi)}{2z_k} \big) \Delta N^-_{t_i}(k) \Big) \bigg)^2 \Delta t
\end{align}

The following is the policy iteration algorithm
\begin{algorithm}
\caption{Policy Iteration Algorithm}
\begin{algorithmic}
\Require Initialize hyperparameters, learning rate $\alpha$, and initial neural networks parameters for policy and value function, $\theta$, $\phi$
\For {l = 1 to L}
    \For{m = 1 to M}
    \State Generate one sample path $\mathcal{D} = \{(t_i, S_{t_i}, q_{t_i})_{i = 1}^T\}$ under policy $\boldsymbol \pi^{\phi}$
\EndFor
    \State Compute $\textbf{ML}(\theta)$
    \State Updates $\theta \leftarrow \theta - \alpha \nabla_{\theta} \textbf{ML}(\theta)$
    \State Update $\boldsymbol \pi^{\phi} \leftarrow \mathcal{N}\Big(\hspace{0.1cm} \boldsymbol \epsilon \hspace{0.1cm} \big | \bigg(\frac{A_k}{2B_k} - \frac{\mathcal{H}^{\pm}_{\theta}(t, S, q, \boldsymbol \pi^{\phi})}{2z_k}  \bigg), \Sigma \Big)$
\EndFor
\end{algorithmic}
\end{algorithm}

\subsection{Actor-Critic Algorithm}
Given our knowledge that the optimal policy follows a Gaussian distribution, we can consider the following model. We introduce a neural network $V^\theta(t, S, q)$ with input $(t, S, q)$ and outputting a scalar value. Additionally, we model the mean of the policy using a neural network, denoted as $\boldsymbol{\pi}^\phi = \mathcal{N} (\boldsymbol M^\phi(t, S, q), \boldsymbol \Sigma)$.

When generating a sample path denoted by $\mathcal{D} = { (t_i, S_{t_i}, q_{t_i}, \boldsymbol \epsilon_{t_i}, \Delta N_{t_i}^+, \Delta N_{t_i}^-)_{i = 0}^{T} }$, we can consider the temporal-difference error for the value function, which is defined as
\begin{align}
    \delta^\theta _{t_i} &= r_{t_i} + V^\theta(t_{i + 1}, S_{t_{i + 1}}, q_{t_{i + 1}}) -  V^\theta(t_{i}, S_{t_{i}}, q_{t_{i}})  \nonumber \\
    &= \big(z_1 \Delta N_{t_i}^+(1), z_1 \Delta N_{t_i}^-(1) , ..., z_N \Delta N_{t_i}^+(N) ,z_N \Delta N_{t_i}^-(N) \big) \boldsymbol{M}^\phi(t, S, q)  + \Big [ q_{t_{i+1}}S_{t_{i+1}} - q_{t_i}S_{t_i} \nonumber \\
    & - \delta q_{t_i}^2 - \gamma \big( N(1 + \log 2\boldsymbol \pi) + \frac{1}{2} \log | \boldsymbol{\Sigma} | \big) \Big] \Delta t + V^\theta(t_{i + 1}, S_{t_{i + 1}}, q_{t_{i + 1}}) -  V^\theta(t_{i}, S_{t_{i}}, q_{t_{i}})
\end{align}

The critic loss is 
\begin{align}
    L(\theta) = \frac{1}{2} \sum_{i = 0}^{T - 1} (\delta_{t_i}^\theta)^2 
\end{align}
then the critic loss's gradient is 
\begin{align}
    \nabla_\theta L(\theta) = \sum_{i = 0}^{T - 1} \delta_{t_i}^\theta \nabla_\theta \delta_{t_i}^\theta
\end{align}
the policy gradient is 
\begin{align}
    \nabla_\phi J(\phi) = \sum_{i = 0}^{T - 1} \delta_{t_i}^\theta \nabla_\phi \log \boldsymbol{\pi}^\phi(\boldsymbol \epsilon_{t_i} |t_{i}, S_{t_{i}}, q_{t_{i}}) 
\end{align}

\begin{algorithm}
\caption{Actor-Critic Algorithm}
\begin{algorithmic}
\Require Initialize hyperparameters, learning rates $\alpha$, $\beta$, and for policy $\boldsymbol \pi^{\phi_0}$, and value function $V^{\theta_0}$
\For {l = 1 to L}
    \State Generate one sample path $\mathcal{D} = \{t_i, S_{t_i}, q_{t_i}, \Delta N_{t_i}^+, \Delta N_{t_i}^-  \}$ under policy $\boldsymbol \pi^{\phi}$
    \State Compute $L(\theta)$
    \State Updates $\theta \leftarrow \theta - \alpha \nabla_\theta L(\theta)$
    \State Updates $\phi \leftarrow \phi - \beta \nabla_\phi J(\phi)$
\EndFor
\end{algorithmic}
\end{algorithm}

\section{Numerical Results}
During our numerical analysis, we established specific parameter values to facilitate the evaluation. The trading window was set to $T = 1$, with a trading interval of $dt = 0.01$. The volatility of the mid-price was set at $\sigma = 0.05$, the exploration coefficient was $\gamma = 0.01$, and the penalty coefficient for inventory was $\delta = 0.01$. Additionally, we defined the order size tiers and the coefficient of market order arrival intensity as follows 
\begin{align*}
    &z = [10, 20, 30, 40, 50, 60] \\
    &A = [20, 18, 15, 12, 10, 8] \\
    &B = [1, 1, 1, 1, 1, 1]
\end{align*}

Our investigation starts with numerical studies on the policy improvement algorithm. The algorithm initiates by initializing two neural networks: one dedicated to capturing the bid-ask policy and another designed to approximate the corresponding value function. Through multiple epochs of training, the bid-ask policy network undergoes refinement. Subsequently, it is replaced with the newly trained value network as prescribed in the policy iteration algorithm.

To determine the most suitable neural network structure for the policy and value functions, we employ three different architectures: Multilayer Perceptron (MLP), Convolutional Neural Network (CNN), and CNN with a residual block. These network architectures are utilized as initial models for the policy and value functions. To evaluate the approximation capabilities of these neural networks with respect to the true value function under a specific policy, we compute the martingale loss, as defined in Equation (32), using 10 randomly generated samples per epoch. The value function network is trained for a total of 50 epochs.

Figures 1-3 display the loss functions associated with the various neural network architectures, with the loss range appropriately adjusted for visualization purposes. Furthermore, Figure 4 presents a comparative analysis of these three loss plots, with a common constant utilized to normalize the loss range. It is essential to emphasize that the constant values employed for division differ across the four figures to ensure accurate representation and comparison.

\begin{figure}[!htp]
  \centering
  \begin{minipage}{0.3\textwidth}
    \centering
    \includegraphics[width=\linewidth]{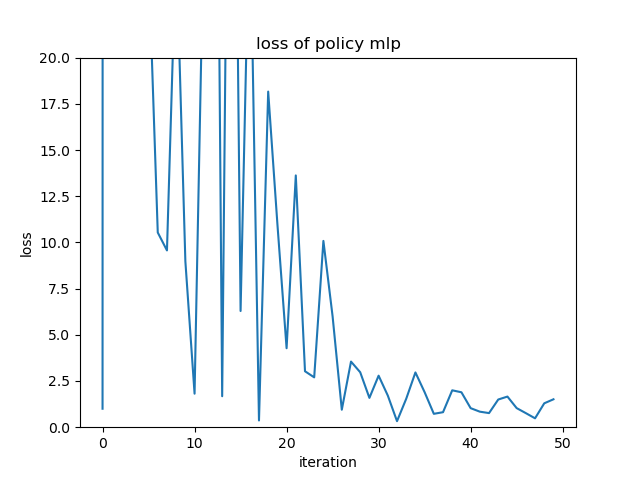}
    \caption{loss of MLP network}
    \label{fig:graph1}
  \end{minipage}
  \hfill
  \begin{minipage}{0.3\textwidth}
    \centering
    \includegraphics[width=\linewidth]{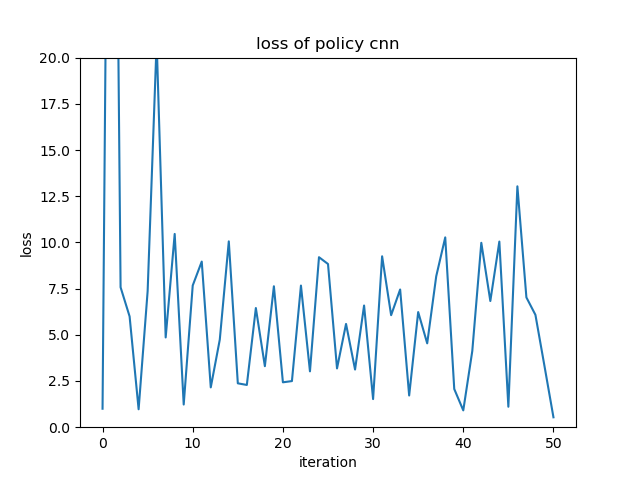}
    \caption{loss of CNN network}
    \label{fig:graph2}
  \end{minipage}%
  \hfill
  \begin{minipage}{0.3\textwidth}
    \centering
    \includegraphics[width=\linewidth]{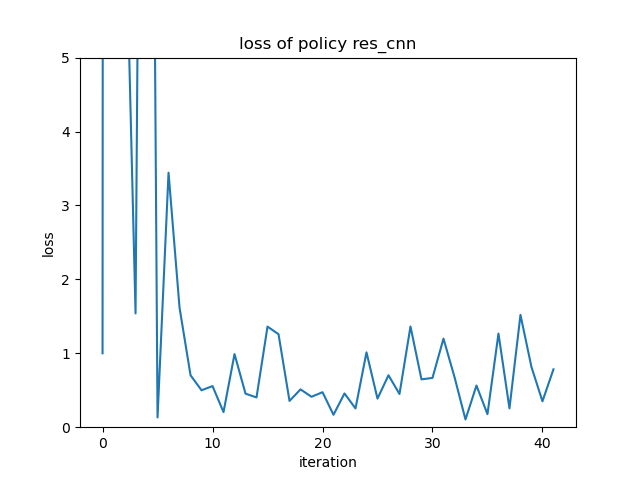}
    \caption{loss of CNN with residual block}
    \label{fig:graph3}
  \end{minipage}
\end{figure}

\begin{figure}[!htp]
    \centering
    \includegraphics[width = 0.5\textwidth]{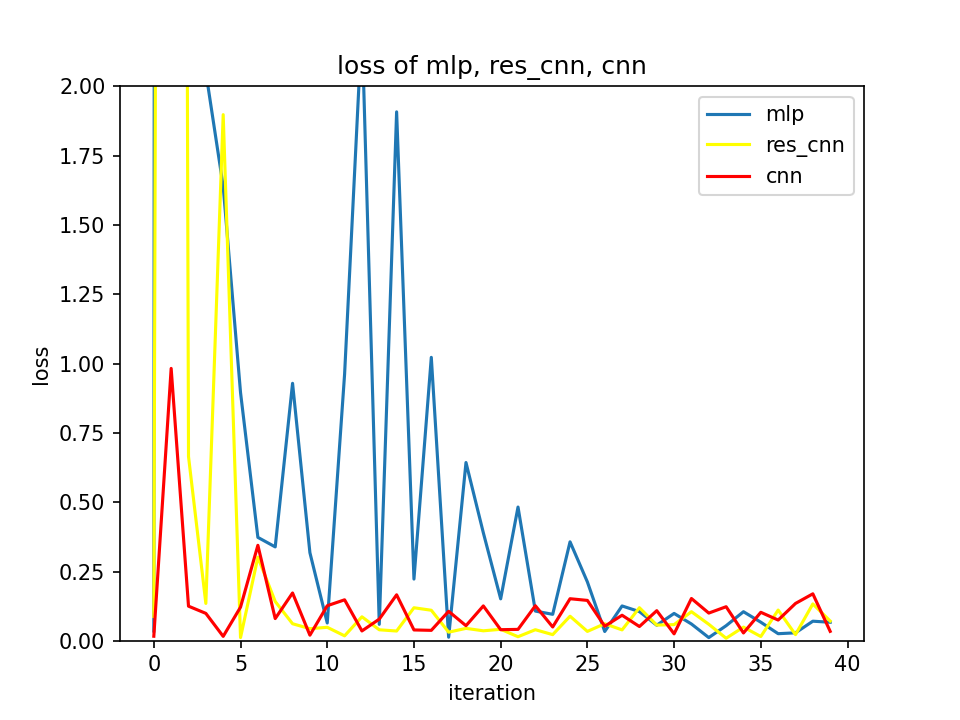}
    \caption{loss of three networks}
    \label{fig:enter-label}
\end{figure}

The loss plots presented above clearly indicate that the CNN with a residual block demonstrates a slightly higher level of approximation capability compared to the CNN without a residual block. Additionally, both CNN architectures outperform the MLP in terms of their ability to approximate the value function. Based on these findings, we have made the decision to utilize the CNN with a residual block for the subsequent numerical studies, as it exhibits the most promising performance in terms of accuracy and approximation.

To model the value function in the policy iteration algorithm and actor-critic algorithm, we utilize a convolutional neural network (CNN) with a residual block. The inputs to the neural network are the variables $(t, S, q)$, and the output is a scalar representing the predicted value.

In the actor-critic algorithm, we employ the same neural network structure to model the mean of the Gaussian policy. However, we modify the output layer to produce a vector representing the mean of the bid-ask spreads.

The CNN architecture consists of 1-dimensional convolutional layers with identical settings. Each convolutional layer has an input channel of 1 and an output channel of 2, with a kernel size of 3 and stride and padding set to 1. The residual block is composed of two convolutional layers with Rectified Linear Unit (ReLU) as the activation function.

The overall structure of the neural network involves mapping the inputs to a higher-dimensional space using a linear layer, which in our case has a dimension of 128. This is followed by two residual blocks and four linear layers with varying dimensions. The activation function used throughout the network is ReLU. By employing this CNN architecture with a residual block, we aim to capture complex patterns and relationships in the input data, allowing for more accurate predictions of the value function and bid-ask spreads.

In our numerical analysis, we conduct a total of five policy iterations using the policy iteration algorithm. At each iteration, we follow the prescribed steps of the algorithm to update the policy. Figures 5-9 showcase the return distributions obtained at each policy iteration. To evaluate the performance of each iteration, we generate a dataset consisting of 100 samples and compute the corresponding returns. These returns are then visualized using histograms, providing a discrete representation of the distribution. To complement the discrete histograms, we employ Gaussian kernel density estimation to derive a continuous representation of the return distributions. This estimation technique allows us to obtain a smooth and continuous description of the observed returns, enabling a more detailed analysis of their characteristics. By examining the return distributions across the five policy iterations, we can gain insights into the evolution of the policy and its impact on the overall returns.

\begin{figure}[!htp]
  \centering
  \begin{minipage}{0.45\textwidth}
    \centering
    \includegraphics[width=\linewidth]{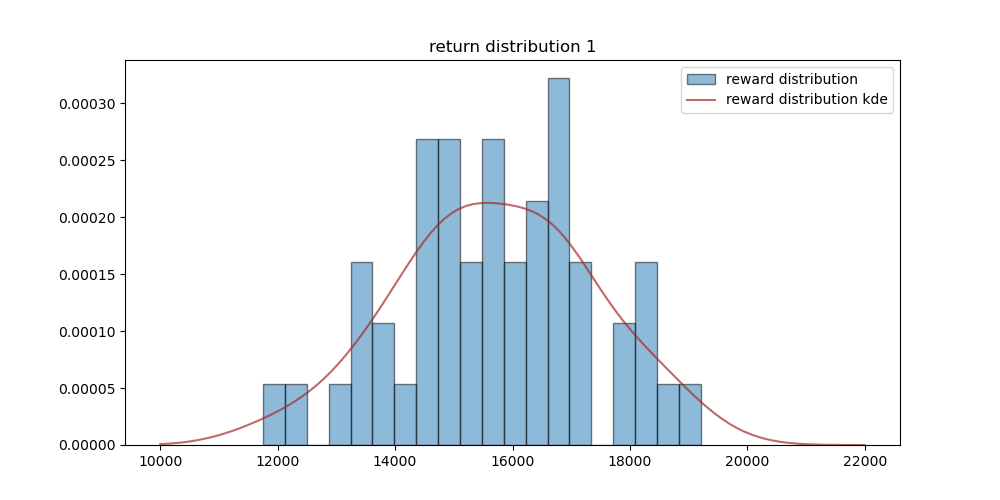}
    \caption{return distribution of after first policy iteration}
    \label{fig:graph1}
  \end{minipage}
  \hfill
  \begin{minipage}{0.45\textwidth}
    \centering
    \includegraphics[width=\linewidth]{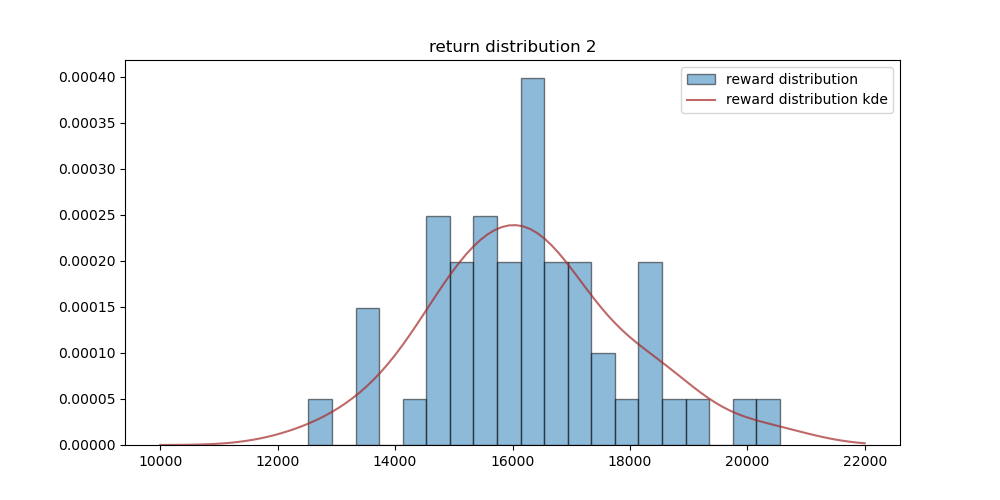}
    \caption{return distribution of after second policy iteration}
    \label{fig:graph2}
  \end{minipage}%
\end{figure}

\begin{figure}[!htp]
  \centering
   \begin{minipage}{0.45\textwidth}
    \centering
    \includegraphics[width=\linewidth]{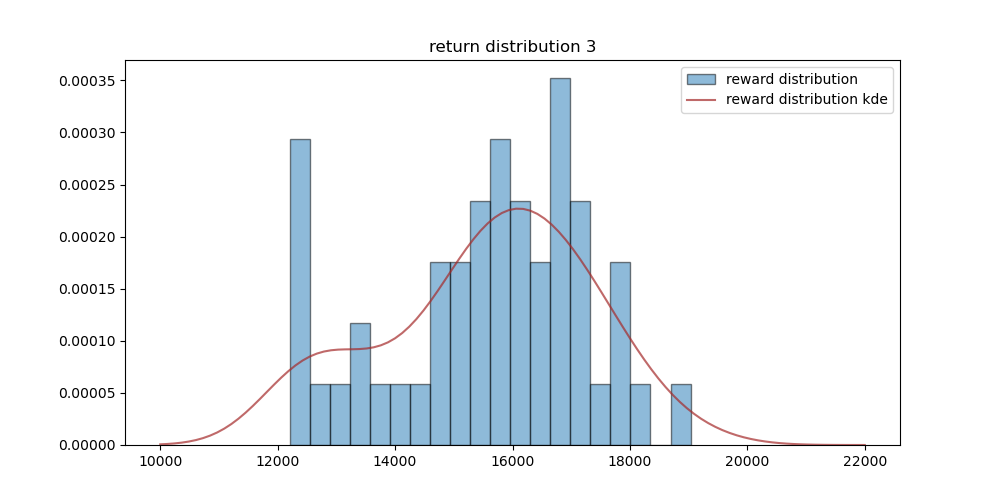}
    \caption{return distribution of after third policy iteration}
    \label{fig:graph3}
  \end{minipage}
  \hfill
  \begin{minipage}{0.45\textwidth}
    \centering
    \includegraphics[width=\linewidth]{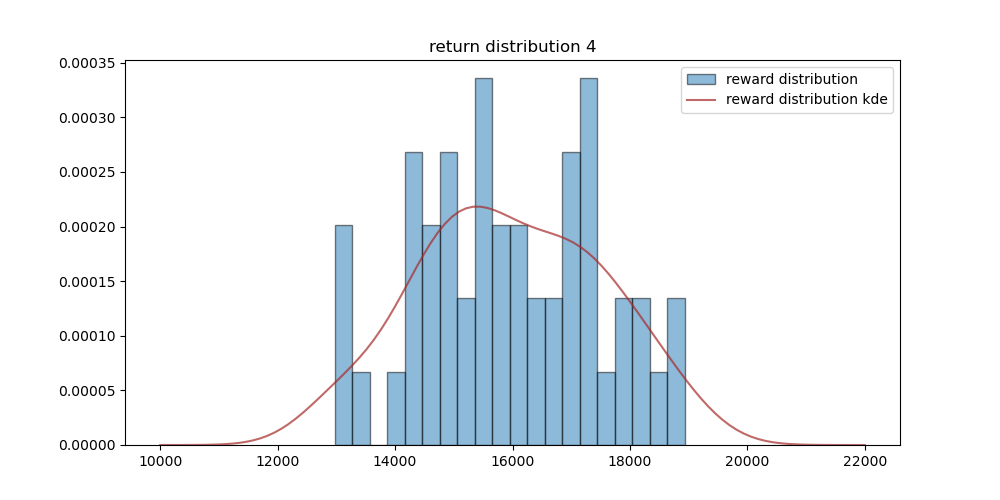}
    \caption{return distribution of after fourth policy iteration}
    \label{fig:graph1}
  \end{minipage}
\end{figure}

\begin{figure}
    \centering
    \includegraphics[width=0.45\textwidth]{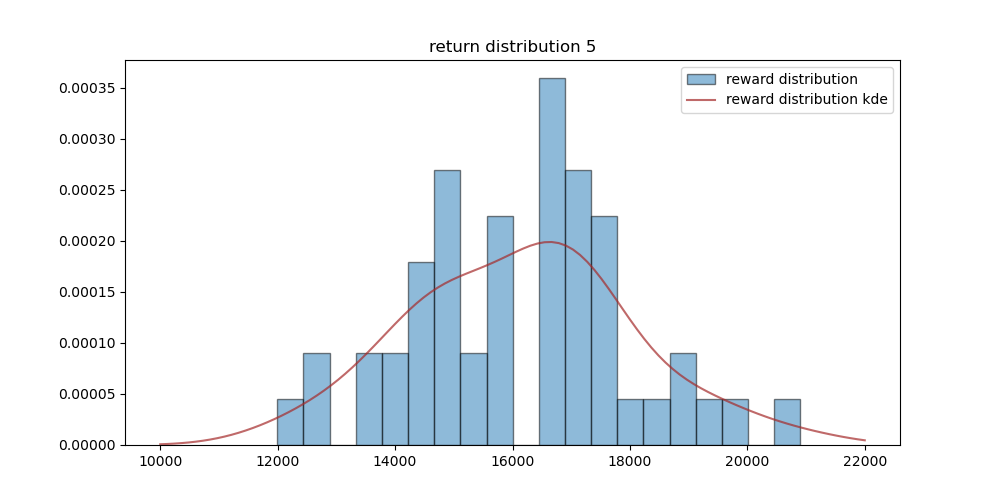}
    \caption{return distribution of after fifth policy iteration}
    \label{fig:graph1}
\end{figure}

Figure 10 presents a comparison of the kernel density estimation (KDE) plots representing the return distributions across the five policy iterations. This comparison allows us to observe the changes in the return distributions as the policy iteration procedure progresses. Furthermore, in Figure 11, we specifically focus on the KDE plots of the return distributions from the first policy iteration and the fifth policy iteration. By examining these specific iterations, we can gain a more detailed understanding of how the return distribution evolves over the course of the policy iteration algorithm. From Figure 11, we can observe that the overall return distribution shows improvement as we progress through the policy iteration procedure. This suggests that the policy iteration algorithm is effective in refining the policy and enhancing the performance in terms of returns. 

It is crucial to acknowledge that the observed returns may exhibit unusually high values, and one potential explanation for this occurrence is the unrealistic magnitude of the bid-ask spreads. To provide further insight into this matter, Figures 12 and 13 showcase the bid-ask spreads at time $t = 0$ for tier 1 and tier 2 securities, respectively, across different reference prices and inventory levels. By examining these figures, we can gain a better understanding of the spread values in different scenarios. Furthermore, Figures 14 and 15 illustrate the bid and ask prices for all tiers at time $t = 0$, with a reference price of $S = 1$ and an inventory level of $q = 50$, which do fit the intuition that larger size order will get narrower bid-ask spreads. These figures offer a visual representation of the bid and ask prices for various tiers and can contribute to our understanding of OTC market making. 

\begin{figure}[!htp]
  \centering
   \begin{minipage}{0.45\textwidth}
    \centering
    \includegraphics[width=\linewidth]{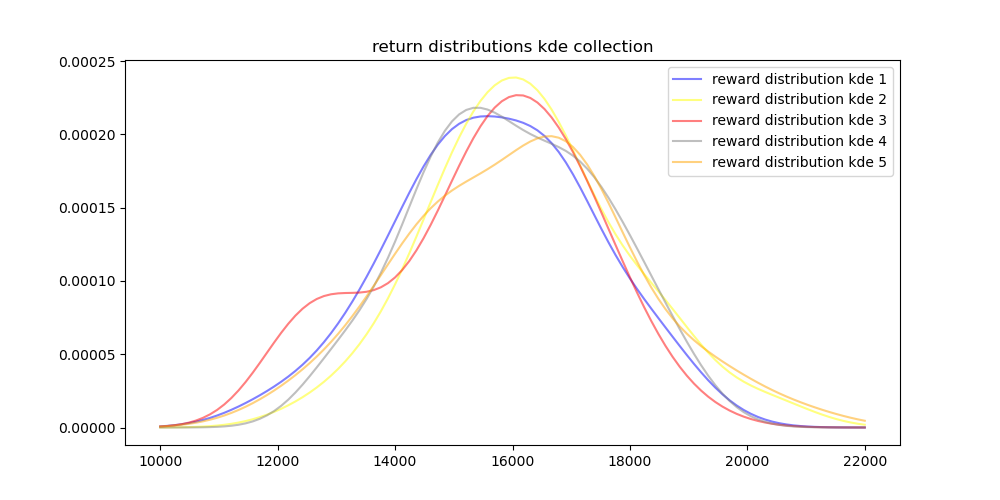}
    \caption{Gaussian kernel density estimations for return distributions of five policy iterations}
    \label{fig:graph3}
  \end{minipage}
  \hfill
  \begin{minipage}{0.45\textwidth}
    \centering
    \includegraphics[width=\linewidth]{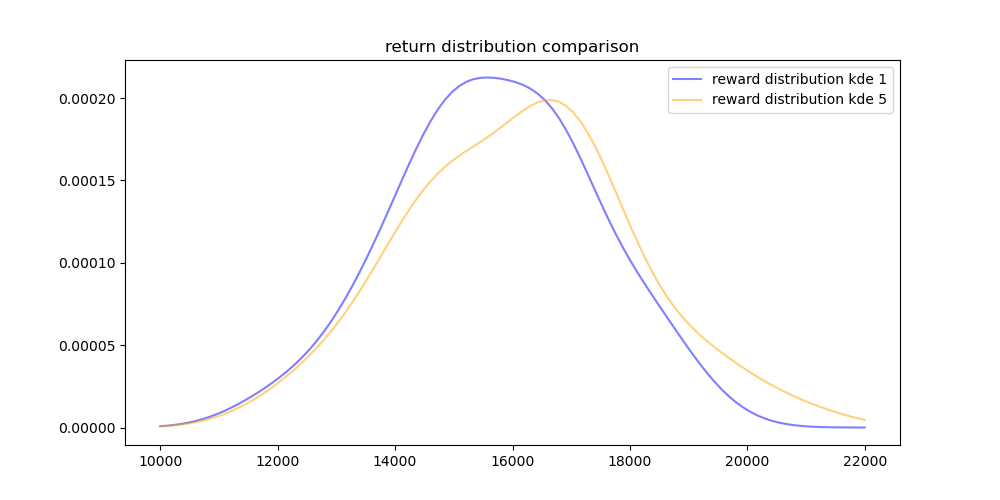}
    \caption{comparison between first and fifth policy iteration's return distribution}
    \label{fig:graph1}
  \end{minipage}
\end{figure}

\begin{figure}[!htp]
  \centering
   \begin{minipage}{0.45\textwidth}
    \centering
    \includegraphics[width=\linewidth]{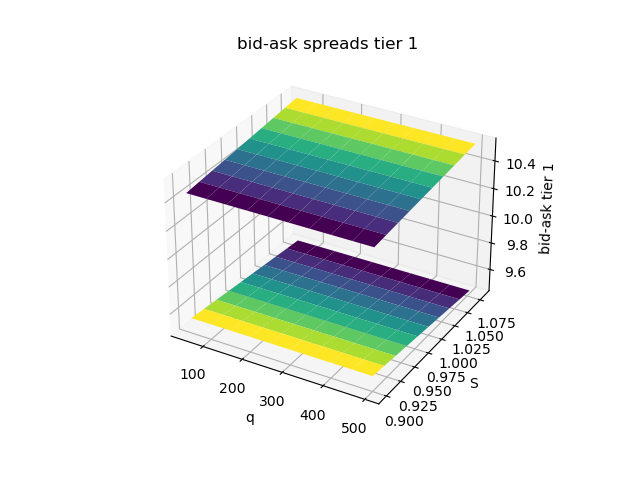}
    \caption{tier 1 bid-ask when t = 0}
    \label{fig:graph3}
  \end{minipage}
  \hfill
  \begin{minipage}{0.45\textwidth}
    \centering
    \includegraphics[width=\linewidth]{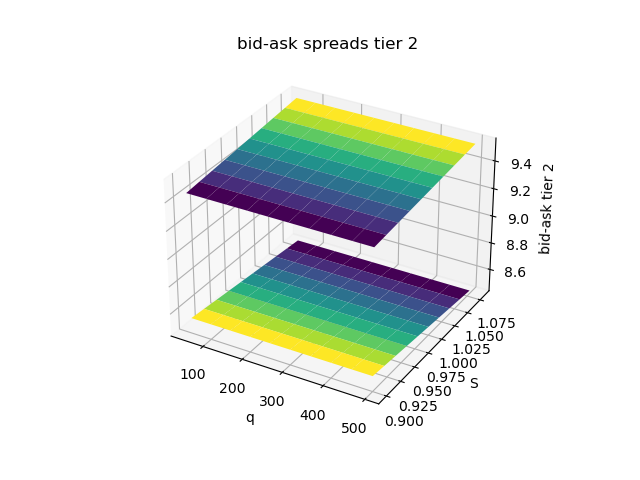}
    \caption{tier 2 bid-ask when t = 0}
    \label{fig:graph1}
  \end{minipage}
\end{figure}

\begin{figure}[!htp]
  \centering
   \begin{minipage}{0.45\textwidth}
    \centering
    \includegraphics[width=\linewidth]{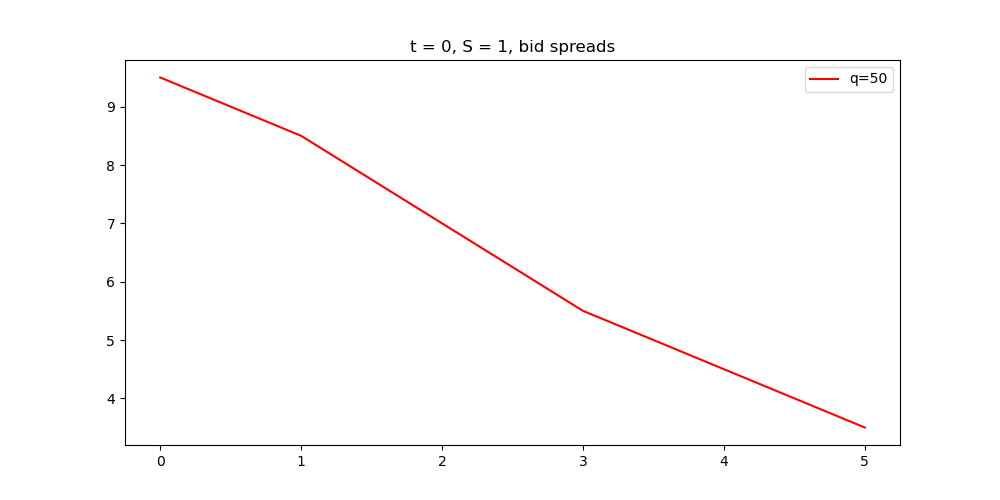}
    \caption{bid spreads when t = 0, S = 1, q = 50}
    \label{fig:graph3}
  \end{minipage}
  \hfill
  \begin{minipage}{0.45\textwidth}
    \centering
    \includegraphics[width=\linewidth]{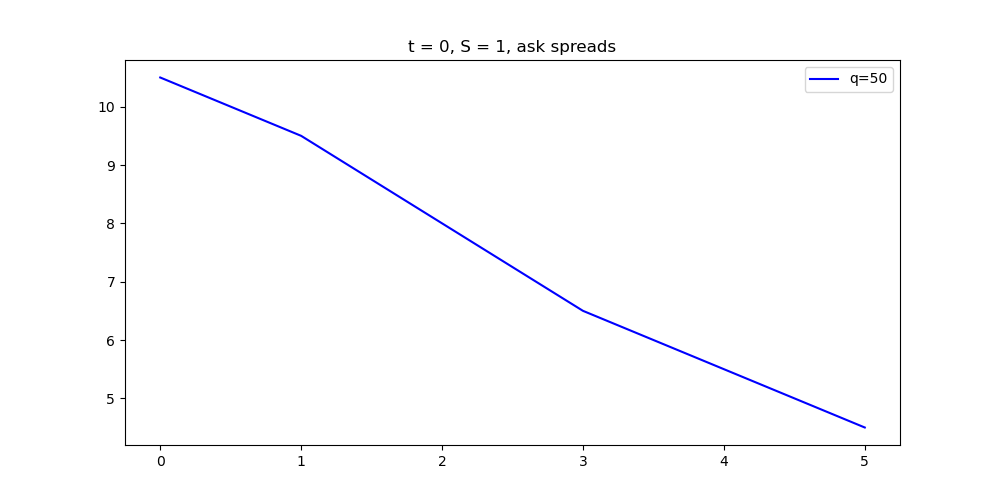}
    \caption{ask spreads when t = 0, S = 1, q = 50}
    \label{fig:graph1}
  \end{minipage}
\end{figure}

Extensive analysis has been conducted to carefully examine this phenomenon. One of the primary reasons for the policy iteration algorithm yielding inflated bid-ask spreads lies in the assumption that the arrival intensity of market orders is linearly dependent on the bid-ask spreads. This modeling assumption becomes unrealistic when confronted with excessively large bid-ask spreads. Consequently, the algorithm can still operate, albeit with potential complications arising from negative bidding prices. Such issues can have profound implications.

Moreover, the policy iteration algorithm has limitations in controlling bid-ask spreads due to the approximation capabilities of neural networks and a limited number of inputs. In contrast, the actor-critic algorithm addresses this challenge by employing a neural network to model the mean output of the bid-ask policy, allowing for better control over bid-ask spreads. This advantage is attributed to the use of suitable activation functions, which constrain the proposed spreads within a reasonable range. Additionally, the actor-critic algorithm offers faster training procedures compared to the policy iteration algorithm. These benefits make the actor-critic algorithm a promising approach for bid-ask spread control in practice.

Figures 16 and 17 provide an insightful depiction of the critic loss and policy loss, respectively, as defined in the actor-critic algorithm.

\begin{figure}[!htp]
  \centering
   \begin{minipage}{0.45\textwidth}
    \centering
    \includegraphics[width=\linewidth]{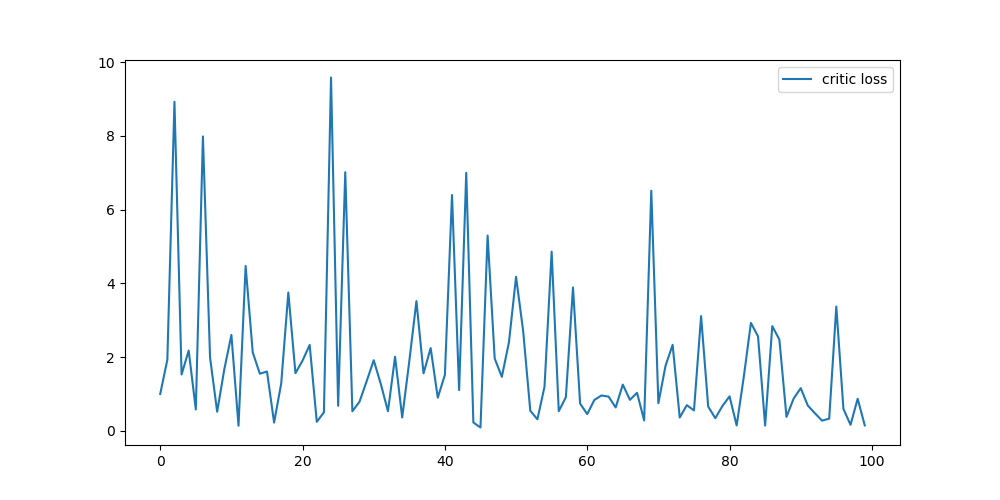}
    \caption{critic loss of actor-critic algorithm}
    \label{fig:graph3}
  \end{minipage}
  \hfill
  \begin{minipage}{0.45\textwidth}
    \centering
    \includegraphics[width=\linewidth]{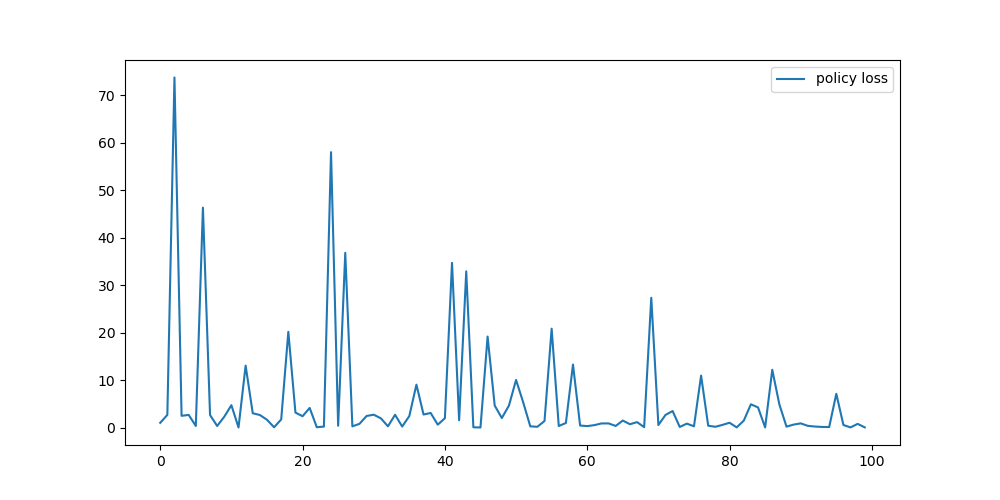}
    \caption{policy loss of actor-critic algorithm}
    \label{fig:graph1}
  \end{minipage}
\end{figure}

\begin{figure}[!htp]
  \centering
   \begin{minipage}{0.45\textwidth}
    \centering
    \includegraphics[width=\linewidth]{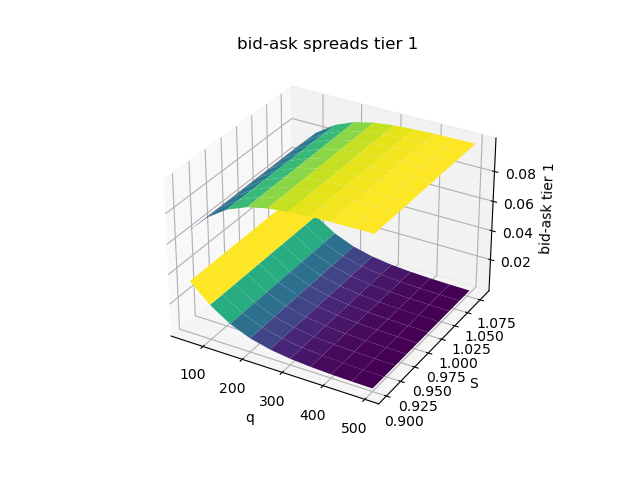}
    \caption{bid-ask spreads of tier 1}
    \label{fig:graph3}
  \end{minipage}
  \hfill
  \begin{minipage}{0.45\textwidth}
    \centering
    \includegraphics[width=\linewidth]{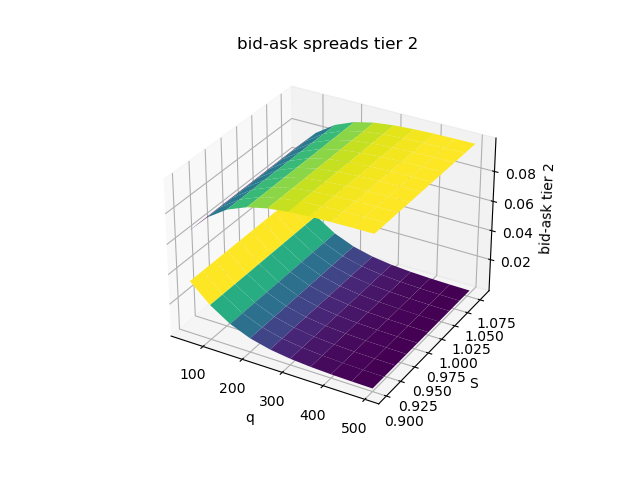}
    \caption{bid-ask spreads of tier 2}
    \label{fig:graph1}
  \end{minipage}
  \hfill
  \begin{minipage}{0.45\textwidth}
    \centering
    \includegraphics[width=\linewidth]{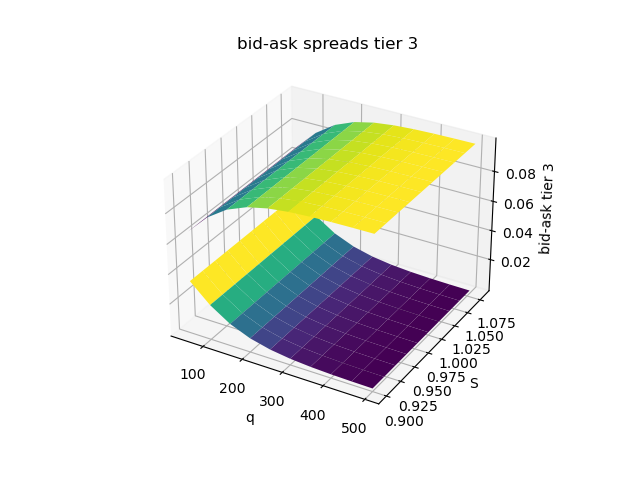}
    \caption{bid-ask spreads of tier 3}
    \label{fig:graph1}
  \end{minipage}
\end{figure}

\begin{figure}
    \centering
    \includegraphics[width = \textwidth]{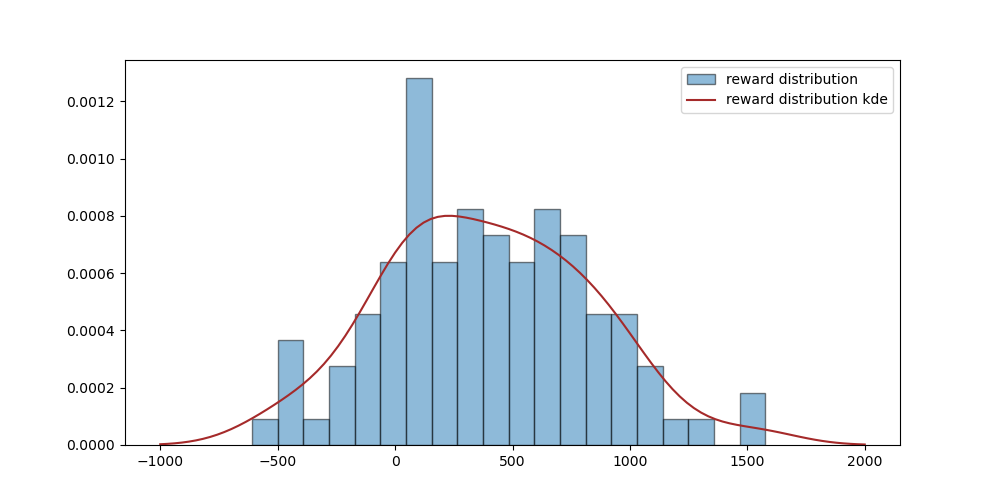}
    \caption{reward distribution actor-critic algorithm}
    \label{fig:enter-label}
\end{figure}

Figures 18-20 visually demonstrate the bid-ask spreads at time $t = 0$ for tiers 1, 2, and 3, respectively. These figures highlight the improved realism and coherence of the bid-ask spreads compared to earlier iterations. Notably, the bid spreads tend to increase as the inventory level rises, while the ask spreads decrease accordingly. This behavior aligns with the goal of effectively managing inventory and reflects the dynamics of a more realistic market-making strategy. The observed patterns validate the efficacy of the refined bid-ask spread formulation in capturing the expected market behavior.

Figure 21 displays the return distribution obtained from the actor-critic algorithm. While the returns in this distribution may not exhibit the same level of exceptional performance as those in the policy iteration algorithm, they demonstrate a closer alignment with the expectations of a realistic market-making strategy. It is important to note that although there are instances of negative returns, the overall profitability of the market-making strategy is evident. This observation further validates the efficacy of the actor-critic algorithm in generating a practical and profitable approach to market-making.

\section{Conclusion}
In our future research, we recognize the importance of going beyond the assumptions made in our current study. To achieve this, we plan to explore non-parametric methods that offer greater flexibility and adaptability to real market data. These methods will allow us to capture the complexities of actual market dynamics and uncover insights that better align with real-world scenarios. By delving into non-parametric techniques, we expect to refine our understanding of the market-making problem and generate more nuanced and robust results.

In summary, this paper introduces a novel reinforcement learning framework specifically tailored for addressing the multi-dimensional market-making problem in OTC markets. By utilizing a stochastic policy and investigating the relationship between market order arrivals and bid-ask spreads, we demonstrate the effectiveness of our proposed approach through extensive numerical experiments. Moving forward, our research efforts will focus on incorporating non-parametric methods to enhance the adaptability of our framework to real-world market data.

\newpage
\bibliography{Reference.bib}
\bibliographystyle{apacite}

\end{document}